\documentclass[11pt]{article}

\usepackage{amsmath,amssymb,amsfonts,amsthm,verbatim}
\usepackage{fullpage}

\usepackage{hyperref}

 \providecommand{\F}{\mathbb{F}}

\newtheorem{lemma}{Lemma}[section]
\newtheorem{theorem}[lemma]{Theorem}
\newtheorem{prop}[lemma]{Proposition}
\newtheorem{cor}[lemma]{Corollary}

\newtheorem{defn}{Definition}
\newtheorem{example}{Example}

\newcommand{\eps}{\varepsilon}
\renewcommand{\epsilon}{\varepsilon}
\renewcommand{\le}{\leqslant}
\renewcommand{\ge}{\geqslant}

\newcommand{\vnote}[1]{}
\newcommand{\cpnote}[1]{}

\def\ZZ{\mathbb{Z}}

\def\PP{\mathbb{P}}

\def \mL {\mathcal{L}}

\def \mP {\mathcal{P}}

\def \mT {\mathcal{T}}

\def\Pin{{P_{\infty}}}
\def\Supp{{\rm Supp}}
\def\wt{{\rm wt}}
\def\RM{{\mathcal RM}}

\newcommand{\Ga}{\alpha}
\newcommand{\Gb}{\beta}
\newcommand{\Gg}{\gamma}     
     
\newcommand{\Ge}{\epsilon}

\def\RM{\mathsf{RM}}

\def \bc {{\bf c}}

\def \bu {{\bf u}}
\def \bv {{\bf v}}
\def \bo {{\bf 0}}

\def\Supp {{\rm Supp }}

\date{}
\begin{document}

\title{\bf Efficiently list-decodable punctured Reed-Muller codes}

\author{Venkatesan Guruswami\thanks{Computer Science Department, Carnegie Mellon University, Pittsburgh, USA. {\tt guruswami@cmu.edu}. Research supported in part by US National Science Foundation grants CCF-1422045 and CCF-1563742.} \and Lingfei Jin\thanks{School of Computer Science, Shanghai Key Laboratory of Intelligent Information Processing, Fudan University, Shanghai 200433, China. {\tt lfjin@fudan.edu.cn}} \and Chaoping Xing\thanks{Division of Mathematical Sciences, School of Physical \&  Mathematical Sciences, Nanyang Technological University, Singapore. {\tt xingcp@ntu.edu.sg}.}}

\maketitle
\thispagestyle{empty}

\begin{abstract}
The Reed-Muller (RM) code encoding $n$-variate degree-$d$ polynomials
over $\F_q$ for $d < q$, with its evaluation on $\F_q^n$, has relative distance $1-d/q$ and can be list
decoded from a $1-O(\sqrt{d/q})$ fraction of errors. In this work, for
$d \ll q$, we give a length-efficient puncturing of such codes which
(almost) retains the distance and list decodability properties of the
Reed-Muller code, but has much better rate.  Specificially, when $q
=\Omega( d^2/\eps^2)$, we given an explicit rate $\Omega\left(\frac{\Ge}{d!}\right)$
puncturing of Reed-Muller codes which have relative distance at least
$(1-\eps)$ and efficient list decoding up to $(1-\sqrt{\eps})$ error
fraction. This almost matches the performance of random puncturings
which work with the weaker field size requirement $q= \Omega(
d/\eps^2)$.  We can also improve the field size requirement to the
optimal (up to constant factors) $q =\Omega( d/\eps)$, at the expense of
a worse list decoding radius of $1-\eps^{1/3}$ and rate
$\Omega\left(\frac{\Ge^2}{d!}\right)$.

The first of the above trade-offs is obtained by substituting for the
variables functions with carefully chosen pole orders from an
algebraic function field; this leads to a puncturing for which the RM
code is a subcode of a certain algebraic-geometric code (which is
known to be efficiently list decodable). The second trade-off is
obtained by concatenating this construction with a Reed-Solomon based
multiplication friendly pair, and using the list recovery property of
algebraic-geometric codes.

\end{abstract}

\section{Introduction}

The Reed-Muller code is one of the oldest and most widely studied code
families, with many fascinating properties.  For a finite field $\F_q$
with $q$ elements, and integers $n,d$, the Reed-Muller  code
$\RM_q(n,d)$ encodes polynomials in $n$-variables of total degree at most $d$
by their evaluations at all points in $\F_q^n$. In the case where $n=1$, the Reed-Muller  code
$\RM_q(1,d)$ is in fact the extended Reed-Solomon code.  In this paper we focus
on the regime where $d < q$ --- this case is of particular interest in
complexity theory, where the local checkability/decodabiliy properties
of the code (thanks to its restrictions to lines/curves being
Reed-Solomon codes) have found many applications. In this regime, the
Reed-Muller code $\RM_q(n,d)$ is an $\F_q$-linear code of dimension ${{n+d}\choose d}$,
block length $q^n$, and relative distance $1-d/q$. Thus, when $d \ll
q$, the distance property is excellent and any two codewords differ in
most of the positions. However, for large $n$, the rate of the code is
very small, as $\approx n^d$ symbols are encoded into $q^n$ codeword
symbols.

The poor rate can be ameliorated, without significant compromise in
the distance properties, by  {\em puncturing} of the code. This
means that message polynomials are encoded by their values at a
 subset $S \subseteq \F_q^n$ of evaluation
points. Standard probabilistic arguments show that if one punctures a $q$-ary $[N,K,D]$  code at $N'$ random positions to obtain an $[N',K,D']$ code, then $D'\approx D(N'-K)/(N-K)$. Thus, if we puncture a Reed-Muller code with $D/N=1-\Ge/2$, down to rate $\Omega(\Ge)$, i.e., $K/N'=\Omega(\Ge)$, then $D'/N' \ge 1-\eps$ since $K\ll N$.
Note
that this makes the rate $\Omega(\eps)$ which is optimal up to
constant factors for a relative distance of $1-O(\eps)$. Such a subset $S$
used for puncturing is called a ``dense hitting set'' for
degree-$d$ polynomials (which means that any non-zero polynomial
vanishes on at most an $\eps$ fraction of $S$). Explicit constructions
of hitting sets (and their stronger variant, pseudorandom generators)
for low-degree polynomials have been well-studied in the literature~\cite{KS01,Bo05,BV,Lov,Vio,DS11,Lu12}, with recent works achieving both field size and seed length optimal up to constant factors~\cite{GX14,B12}.

In other words, when $q =\Omega(d/\eps)$, we know explicit puncturings
of $\RM_q(n,d)$ to length $\text{poly}(n^d/\eps)$ that have relative
distance $1-\eps$. The motivation in this work is to obtain such
length-efficient puncturings of an RM code, which in addition have
{\em efficient error-correction algorithms}. Our goal is to correct a
constant fraction of worst-case errors, and in fact to \emph{list decode} an error fraction approaching $1$ for small $\eps$.
Our main result is stated below (we achieve two incomparable trade-offs):

\begin{theorem}
\label{thm:main-intro}
[Informal statement]
Let $\eps > 0$ be a small constant.
There are explicit puncturings of $\RM_q(n,d)$, constructible in deterministic $\mathrm{poly}(n^d/\eps)$ time, that have relative distance at least $(1-\eps)$ along with
the following list-decodability guarantees:
\begin{enumerate}
\item[{\rm (i)}] When $q = \Omega(d^2/\eps^2)$, the code has rate $\Omega\left(\frac{\Ge}{d!}\right)$ and can be efficiently list decoded from $(1-\sqrt{\eps})$ fraction of errors;
\item[{\rm (ii)}] When $q =\Omega(d/\eps)$, the code has rate $\Omega\left(\frac{\Ge^2}{d!}\right)$ and can be efficiently list decoded from $(1-\sqrt[3]{\eps})$ fraction of errors.
\end{enumerate}
In both cases the output list size is $O(1/\eps)$.
\end{theorem}
 Note that the advantage of the second result of the above theorem is the smaller field size, which is optimal up to constant factors for a relative distance of $1-\eps$.

As a comparison to the {\em combinatorial} (not efficient) decodability of punctured RM codes, for $q \ge \Omega(d/\eps)$, a random puncturing of $\RM_q(n,d)$ of rate $\Omega(\eps)$ is list-decodable from an error-fraction of $(1-\sqrt{\eps})$ by virtue of the Johnson bound on list decoding. Our guarantee (i) above matches this decoding radius with an efficient algorithm under a quadratically larger field size requirement. The guarantee (ii) on the other hand gets the correct field size (for the desired relative distance of $(1-\eps)$), but has worse list decoding radius and rate.
 Specializing a general result of Rudra and Wootters~\cite{RW14}
on list-decodability of randomly punctured codes to Reed-Muller codes shows
that most puncturings of $\RM_q(n,d)$ of rate $\Omega\left(\frac{\eps}{\log^5 (1/\eps) \log q}\right)$ are list-decodable from
an error-fraction $(1-\eps)$, provided $q \ge
\Omega(d/\eps^2)$. Compared to this result, our rate in guarantee (i) and field size in guarantee (ii) are better, but both trade-offs in Theorem~\ref{thm:main-intro} are worse in terms of list decoding radius.
Going past the $1-\sqrt{\eps}$ error fraction for \emph{efficient} list decoding of a rate $\approx \eps$ puncturing remains a challenging open problem, even for the case of Reed-Solomon codes.

Efficiently decodable puncturings of Reed-Muller codes were
studied in \cite{DS11} under the label of noisy interpolating
sets. The authors of \cite{DS11} gave an explicit puncturing to a set
of size $C_d n^d$, for some constant $C_d$ depending on $d$, along with an efficient algorithm to
correct a $\exp(-O(d))$ fraction of errors, which approaches $0$ as
$d$ increases. Our focus is on the list decoding regime where the goal
is to correct a large fraction (approaching $1$) of errors.

\smallskip
\noindent {\bf Approach in brief.} We give a brief description of the
high level approach behind our two constructions claimed in
Theorem~\ref{thm:main-intro}. The idea behind (i) is to replace each
variable $x_i$, $i=1,2,\dots,n$, by a function $f_i$ from a suitable
algebraic function field $K$ that has $a_i$ poles at a specific point
$P_\infty$ and no poles elsewhere. The $a_i$'s are chosen to belong to
a Sidon sequence so that no two (multi) subsets of $\{a_1,\dots,a_n\}$
of size at most $d$ have the same sum. This ensures that every
polynomial of degree $d$ in $x_1,x_2,\dots,x_n$ is mapped to a
distinct function in $K$ with bounded pole order at $P_\infty$. The
evaluations of the $f_i$'s at enough rational points gives the desired
puncturing, which inherits the distance and list decodability
properties from a certain algebraic-geometric code. We note that the
ideas of replacing variables by functions of a function field was
already proposed in \cite{CT13,CGX12}, and replacing $x_i$ by
$x^{a_i}$ for an indeterminate $x$ and $a_i$ from a Sidon set was
proposed in \cite{B12}, and here we combine these two ideas to prove (i). For
the second construction, we concatenate the codes from (i) with
Reed-Solomon based multiplication friendly codes, and use the list
decodability of Reed-Solomon codes and the list recoverability of
algebraic-geometric codes to list decode the concatenated code. The idea to use multiplication friendly codes to reduce the field size in Reed-Muller puncturings appeared in the recent works~\cite{B12,GX14}.

\smallskip
\noindent {\bf Organization.}  The paper is organized as follows. In
Section 2, list decodability of random punctured Reed-Muller codes is
discussed in order to compare with our efficient list decoding of
explicit punctured Reed-Muller codes. Some results on Sidon sequences
are presented in Section 3 for the sake of construction of our
explicit punctured Reed-Muller codes. Preliminaries on function fields
and algebraic geometry codes are included in Section 4. In Section 5,
we introduce multiplicative friendly pairs in order to concatenate our
punctured Reed-Muller codes over larger fields. Section 6 is devoted to
explicit constructions punctured Reed-Muller codes. The main results
on efficient list decodability of our explicit punctured Reed-Muller
codes are presented in Section 7.

\section{List decodability of random punctured Reed-Muller codes }
Throughout, we denote by $X$ the variable vector $(x_1,x_2,\dots,x_n)$. Thus, the multivariate polynomial ring $\F_q[x_1,\dots,x_n]$ is denoted by $\F_q[X]$. A polynomial $F(X)$ of $\F_q[X]$ can be written as $F(X)=\sum_Ia_IX^I$ for some $a_I\in\F_q$, where $I=(i_1,i_2,\dots,i_n)\in\ZZ^n_{\ge 0}$ and $X^I$ denotes $x_1^{i_1}x_2^{i_2}\cdots x_n^{i_n}$.

For $n,d,q$ with $d<q$, define the Reed-Muller code
\[\RM_q(n,d):=\{(F(\bu))_{\bu\in\F_q^n}:\; F\in\F_q[X],\;\deg(F)\le d\}.\]
The  code $\RM_q(n,d)$ has length $q^n$ and dimension ${n+d\choose n}$.

A set is called a {\it multiset} if elements in this set could be repeated. For a multiset $\mT=\{\bu_1,\bu_2,\dots,\bu_N\}$ in $\F_q^n$, denote by $\RM_{q}(n,d)|_{\mT}$ the $\F_q$-linear code \[\RM_q(n,d)|_{\mT}:=\{(F(\bu_1),F(\bu_2),\dots,F(\bu_N)):\; F\in\F_q[X],\;\deg(F)\le d\}.\] It is called a {\it punctured Reed-Muller} code. Note that our punctured Reed-Muller code is slightly different from the notion of usual punctured codes where elements in the set $\mT$ are not repeated.

\begin{defn}
A code $C \subseteq \Sigma^n$ is said to be $(\rho,L)$-list decodable if for every $y \in \Sigma^n$, the number of codewords of $C$ which are within Hamming distance $\rho n$ from $y$ is at most $L$.
\end{defn}

To compare our explicit constructions with  non-constructive puncturings of Reed-Muller codes, we state the following result that can be easily derived from Corollary 2 in Sec. 4.1 of \cite{RW14}.

\begin{prop}\label{prop:2.2}
\begin{enumerate}
\item
  When $q\ge \Omega(d/\Ge^2)$, there is a puncturing of $\RM_q(n,d)$ of rate $\Omega\left(\frac{\eps}{\log^5 (1/\eps) \log q}\right)$ that is
  $(1-\Ge,O(1/\Ge))$-list decodable.
\item When $q\ge \Omega(d/\Ge)$, there is a puncturing of $\RM_q(n,d)$ of rate $\Omega(\eps)$ that is $(1-\sqrt{\eps},O(1/\eps))$-list decodable.
\end{enumerate}
\end{prop}

\section{Sidon sequences}\label{sec:sidon}

The notion of Sidon sequences was introduced by Simon Sidon to study Fourier series and later it became a well-known combinatorial problem and found many applications \cite{ET41,B12}. There is an efficient construction of Sidon sequences given in \cite{B12}. In this section,  we slightly improve the sizes of the elements in the Sidon sequences  \cite{B12} for the purpose of our application to efficient construction of punctured Reed-Muller codes.

For a nonnegative integer vector $I=(i_1,i_2,\dots, i_n)\in\ZZ_{\ge 0}^n$, we denote by $\wt(I)$ the sum $\sum_{k=1}^ni_k$.
\begin{defn} For an integer $d\ge 2$, a sequence $\{a_i\}_{i=1}^n$ of nonnegative integers is called $d$-Sidon sequence if, for any two distinct vectors $I=(i_1,i_2,\dots,i_n),J=(j_1,j_2,\dots,j_n)\in\ZZ_{\ge 0}^n$ with $\max\{\wt(I),\wt(J)\}\le d$, one has $\sum_{k=1}^ni_ka_{k}\not=\sum_{k=1}^nj_ka_{k}$.
\end{defn}
For most applications, we need the largest number in a Sidon sequence to be as small as possible. However, for our purpose, we also want a lower bound on the smallest number in the Sidon sequence. This can be ensured using the following lemma.
\begin{lemma}\label{lem:sidon1} If $\{a_i\}_{i=1}^n$ is  a $d$-Sidon sequence, then for any $M> \max_{1\le i\le n}\{a_i\}$, the sequence $\{a_i+dM\}_{i=1}^n$ is also a $d$-Sidon sequence.
\end{lemma}
\begin{proof} Consider two distinct vectors $I=(i_1,i_2,\dots,i_n),J=(j_1,j_2,\dots,j_n)\in\ZZ_{\ge 0}^n$ with
\[ \max\{\wt(I),\wt(J)\}\le d \ . \]
If $\wt(I)=\wt(J)$, then
\[\sum_{k=1}^ni_k(a_{k}+dM)=\sum_{k=1}^ni_ka_{k}+\wt(I)dM\not=\sum_{k=1}^nj_ka_{k}+\wt(J)dM=\sum_{k=1}^nj_k(a_{k}+dM).\]
 If $\wt(I)\not=\wt(J)$, we may assume that $\wt(I)<\wt(J)$. Then
\begin{eqnarray*} \sum_{k=1}^nj_k(a_{k}+dM)&=&\sum_{k=1}^nj_ka_{k}+\wt(J)dM\ge \wt(J)dM\ge dM+\wt(I)dM\\
&>& \sum_{k=1}^ni_ka_{k}+\wt(I)dM=\sum_{k=1}^ni_k(a_{k}+dM).\end{eqnarray*}
This completes the proof.
\end{proof}
An efficient construction of Sidon sequences was proposed in \cite[Lemma 59]{B12} based on finite fields and discrete logarithm. The largest number in the Sidon sequence given  in \cite[Lemma 59]{B12} is upper bounded by $n^{d+1}+O(n^d)$. If we apply this upper bound to our construction, one would get a punctured Reed-Muller code with rate zero (see Section \ref{sec:RM}). Therefore, we modify the construction of Sidon sequences in \cite[Lemma 59]{B12} as below.
\begin{lemma}\label{lem:sidon2} For arbitrary positive integers $d$ and  $n$, a $d$-Sidon sequence $\{a_i\}_{i=1}^n$ with
\[\max_{1\le i\le n}a_i\le n^d+O(n^{d-1})\] can be deterministically constructed in $O(n^{d/2+2})$ time.
\end{lemma}
\begin{proof} Let $p$ be the smallest prime satisfying $p\ge n$. Then we have $p=n+O(n^{0.525})$ (see \cite{BHP01}). Put $M:= (p^{d+1}-1)/(p-1)$. Then $M=n^d+O(n^{d-1})$.
Let $\Ga_1,\dots,\Ga_n$ be $n$ distinct elements of $\F_p$. One can find a primitive element $\Gg$ of $\F_{p^{d+1}}$ in time $O(p^{d/2})=O(n^{d/2})$ (see \cite[Section 11.1]{S09}). For each $i$, find the discrete logarithm $\ell_i<p^{d+1}-1$ such that $\Gg+\Ga_i=\Gg^{\ell_i}\in\F_{p^{d+1}}$ in time $O(p^{d/2+1})=O(n^{d/2+1})$ (see \cite[Section 11.2]{S09}).  Write $\ell_i=tM+c_i$ for some integer $c_i$ with $0\le c_i<M$. Hence, we have $\Gg+\Ga_i=\beta_i\Gg^{c_i}$, where $\Gb_i=(\Gg^M)^t\in\F_p^*$. We claim that $\{c_1,\dots,c_n\}$ is a $d$-Sidon sequence. Indeed, suppose that there were two distinct vectors $I=(i_1,i_2,\dots,i_n),J=(j_1,j_2,\dots,j_n)\in\ZZ_{\ge 0}^n$ with $\max\{\wt(I),\wt(J)\}\le d$  such that $\sum_{k=1}^ni_kc_{k}=\sum_{k=1}^nj_kc_{k}$. Then we have
$\Gg^{\sum_{k=1}^ni_kc_{k}}=\Gg^{\sum_{k=1}^nj_kc_{k}}$, i.e.,
$\prod_{k=1}^n(\Gg+\Ga_{k})^{i_k}=\Gb \prod_{k=1}^n(\Gg+\Ga_{k})^{j_k}$ for some $\Gb\in\F_{p}^*$. This implies that $\Gg$ is a root of nonzero polynomial $\prod_{k=1}^n(x+\Ga_{k})^{i_k}-\Gb \prod_{k=1}^n(x+\Ga_{k})^{j_k}$ of degree less than $d+1$. This contradiction implies that $\{c_1,\dots,c_n\}$ is a $d$-Sidon sequence.
\end{proof}
Combining Lemma \ref{lem:sidon1} with Lemma \ref{lem:sidon2} gives the following efficient construction of Sidon sequences in a relatively small interval.
\begin{cor}\label{cor:3.3} For arbitrary positive integers $d$ and  $n$, there exists a $d$-Sidon sequence $\{a_i\}_{i=1}^n$ satisfying $dM\le a_i< (d+1)M$ for all $1\le i\le n$, where $M=n^d+O(n^{d-1})$.  Furthermore, this sequence can be constructed in time $O(n^{d/2+2})$.
\end{cor}

\section{Function fields and algebraic geometry codes}\label{sec:FF}

\subsection{Algebraic function fields}\label{subsec:FF}
Let us introduce some background on algebraic function fields over finite fields. The reader may refer to \cite{Stich93,TV91,NX01} for details.

Let $F$ be a function field of one variable over $\F_q$, i.e., $\F_q$ is algebraically closed in $F$ and $F$ is a finite extension over $\F_q(x)$ for any $x\in F\setminus\F_q$.
We denote by $g(F)$ the genus of $F$. In case there is no confusion, we simply use $g$ to denote the genus $g(F)$. Denote by $N(F/\F_q)$ the number of rational  places of $F$.

For a place $Q$ of $F$, we denote by $\nu_Q$ the normalized discrete valuation at $Q$. Denote by $\PP_{F}$  the set of places of $F$. A divisor $G$ of $F$ is a  formal sum $G=\sum_{Q\in\PP_{F}}  n_QQ$, where  $n_Q$ are integers and almost all $n_Q$ are equal to zero.  The degree $\deg(G)$ of $G$ is defined by $\deg(G)=\sum_{Q\in\PP_{F}}  n_Q\deg(Q)$. The support of $G$, denoted by $\Supp(G)$, is defined to be the set $\{Q\in\PP_{F}:\; n_Q\not=0\}$. For two divisors $G=\sum_{Q\in\PP_{F}}  n_QQ$ and $H=\sum_{Q\in\PP_{F}}  m_QQ$,
 we say  that $H$ is less than or equal to $G$, denoted by $H\le G$, if $m_Q\le n_Q$ for all $Q\in\PP_{F}$. For a nonzero function $x$ in $F$, one can define a principal divisor $(x):=\sum_{Q\in\PP_{F}} \nu_Q(x)Q$. It is a well-known fact that the degree of $(x)$ is zero. The zero divisor and pole divisor of $x$ are defined to be $(x)_0:=\sum_{Q\in\PP_{F},\nu_Q(x)>0} \nu_Q(x)Q$ and $(x)_{\infty}:=-\sum_{Q\in\PP_{F},\nu_Q(x)<0} \nu_Q(x)Q$, respectively.

Now for a divisor $G$, one can
define the Riemann-Roch space $\mL(G)$ associated with $G$ by
\begin{equation}\label{eq:4.1} \mL(G):=\{x\in F^*:\; (x)+G\ge 0\}\cup\{0\}.
\end{equation}
Then the Riemann-Roch Theorem says that $\mL(G)$ is a finitely dimensional space over $\F_q$ and its dimension, denoted by $\ell(G)$, is lower bounded by
\begin{equation}\label{eq:4.2} \ell(G)\ge \deg(G)-g+1.
\end{equation}
Furthermore, the equality in (\ref{eq:4.2}) holds if $\deg(G)\ge 2g-1$.
\subsection{Algebraic geometry codes}\label{subsec:AG}
The books \cite{Stich93,TV91,NX01} are standard textbooks for algebraic geometry codes. In this subsection, we briefly introduce the construction of algebraic geometry codes, i.e., Goppa geometric codes.

Let $\mP=\{P_1,P_2,\dots,P_N\}$ be a set of rational places of $F$ and let $G$ be a divisor of $F$ such that $\Supp(G)\cap\mP=\emptyset$. Define the algebraic geometry code by
\begin{equation}\label{eq:4.3}
C(\mP;G)=\{(f(P_1),f(P_2),\dots,f(P_N)):\; f\in\mL(G)\}.
\end{equation}
Then $C(\mP;G)$ is an $\F_q$-linear code of length $N$. The dimension of  $C(\mP;G)$ is $\ell(G)$ as long as $\deg(G)<N$. The minimum distance of the code is at least $N-\deg(G)$ and it is called the designed distance of $C(\mP;G)$.
\subsection{Garcia-Stichtenoth tower}\label{subsec:GS}
There
are two optimal Garcia-Stichtenoth towers that are equivalent. The reader may refer to \cite{GS95,GS96} for
detailed results on the Garcia-Stichtenoth function tower. A low-complexity algorithm for the construction of the Riemann-Roch space and algebraic geometry codes based on the tower of \cite{GS95,GS96} was given in \cite{SAKSD01}.

Let $r$ be a prime power.  The Garcia-Stichtenoth tower is defined by the following recursive
equations \cite{GS96}
\begin{equation}\label{eq:4.4} z_{i+1}^r+z_{i+1}=\frac{z_i^{r}}{z_i^{r-1}+1},\quad i=1,2,\dots,t-1.\end{equation}
Put  $F_t=\F_{r^2}(z_1,z_2,\dots,z_{t})$ for $t\ge 2$.

The function field $F_t$ has at least $r^{t-1}(r^2-r)+1$ rational places. One of these is the ``point at infinity" which is the unique pole $\Pin$ of $z_1$.
The genus $g_t$ of the function field $F_t$ is given by
\[g_t=\left\{\begin{array}{ll}
(r^{t/2}-1)^2&\mbox{if $t$ is even}\\
(r^{(t-1)/2}-1)(r^{(t+1)/2}-1)&\mbox{if $t$ is odd.}\end{array}
\right.\]
Thus, the genus $g_t$ is at most $r^t-r^{t/2}$. Furthermore, the algebraic geometry codes constructed from the Garcia-Stichtenoth tower can be constructed in time cubic in the block length (see \cite{SAKSD01}).

\section{Multiplication friendly pairs}
Multiplication friendly pairs were first introduced by  D.V. Chudnovsky and G.V. Chudnovsky \cite{CC87} to study the complexity of
bilinear multiplication algorithms in extension fields.   Following their work, Shparlinski, Tsfasman and Vl\u{a}du\c{t} \cite{STV92} systematically studied this idea
and extended the result in \cite{CC87}. Recently, multiplication friendly pairs were used to study multiplicative secret sharing \cite{CCX11} and punctured Reed-Muller codes (or equivalently, hitting sets)  \cite{B12,GX14}.

For $d$ vectors $\bc_i=(c_{i1},c_{i2},\dots,c_{im})\in\F_q^m$ ($i=1,2,\dots,d)$, we denote by $\bc_1\ast\bc_2\ast\cdots\ast\bc_d$ the coordinate-wise product $(\prod_{i=1}^dc_{i1},\prod_{i=1}^dc_{i2},\dots,\prod_{i=1}^dc_{im})$. In particular, we denote by $\bc^{\ast d}$ the vector $(c_1^d,\cdots,c_n^d)$ if $\bc=(c_1,\dots,c_n)$.

For an $\F_q$-linear code $C$, we denote by $C^{\ast d}$ the linear code\[{\rm Span}\{\bc_1\ast\bc_2\ast\dots\ast\bc_d:\; \bc_1,\bc_2,\dots,\bc_d\in C\}.\]

\begin{defn}\label{def:2} For $m\ge k$, a pair $(\pi, \psi)$ is called a $(d,k,m)_q$-multiplication friendly pair if $\pi$ is an $\F_q$-linear map from $\F_{q^k}$ to $\F_q^m$ and $\psi$ is   an $\F_q$-linear map from $\F_q^m$  to $\F_{q^k}$ such that $\pi(1)=(1,\dots,1)$ and $\psi(\pi(\Ga_1)\ast\cdots\ast\pi(\Ga_d))=\Ga_1\cdots\Ga_d$ for all $\Ga_i\in\F_{q^k}$. A $(2,k,m)_q$-multiplication friendly pair is also called a bilinear multiplication friendly pair.
\end{defn}

The following two lemmas can be found in \cite{CC87,GX14}.

\begin{lemma}\label{lem:5.1} If $(\pi, \psi)$  is  a $(d,k,m)_q$-multiplication friendly pair, then $\pi$ is injective.
\end{lemma}
\begin{proof} Suppose that $\pi$ is not injective. Then there exists a nonzero element $\Ga\in\F_{q^k}$ such that $\pi(\Ga)=\bo$. Hence $0=\psi(\bo)=\psi(\pi(\Ga)\ast\pi(1)\ast\cdots\ast\pi(1))=\Ga\cdot1\cdots\cdot 1=\Ga$, which is a contradiction.
\end{proof}

\begin{lemma}\label{lem:5.2} If $d,m,k,q$ satisfy $d\ge 2$ and $d(k-1)<m\le q$, then a  $(d,k,m)_q$-multiplication friendly pair $(\pi,\psi)$ such that $(\pi(\F_{q^k}))^{\ast d}$ is a $q$-ary linear code with minimum distance at least $m-d(k-1)$ can be explicitly constructed in polynomial time.
\end{lemma}
\begin{proof} Consider the polynomial ring  $\F_q[x]$.
   Choose an irreducible polynomial $Q(x)$ of degree $k$ and identify $\F_{q^k}$ with the residue field $\F_q[x]/(Q(x))$.  Let $\Ga_1,\Ga_2,\dots,\Ga_m$ be distinct elements of $\F_q$.

For a polynomial $f(x)\in \F_q[x]$, we denote by $\overline{f(x)}$ (or simply $\bar{f}$) to denote the residue of $f(x)$ in the residue field $\F_q[x]/(Q(x))$. Define an $\F_q$-linear map $\pi$ from $\F_{q^k}=\F_q[x]/(Q(x))$ to $\F_q^m$ given by $\pi(\bar{z})=(z(\Ga_1),z(\Ga_2),$ $\dots,z(\Ga_m))$ for an element $\bar{z}\in \F_q[x]/(Q(x))$ with $\deg(z(x))\le k-1$. This is injective since $k< m$ and $\pi(1)=(1,1,\dots,1)$.

Let $C$ be the Reed-Solomon code given by $\{(f(\Ga_1),\dots,f(\Ga_m)):\; f\in\F_q[x],\ \deg(f)\le d(k-1)\}$. Since $d(k-1)<m$, $C$ is a $q$-ary $[m,d(k-1)+1,m-d(k-1)]$-linear code in $\F_q^m$. For a codeword $\bc\in C$, we denote it by $\bc_f$ if $\bc=(f(\Ga_1),\dots,f(\Ga_m))$. Thus, we can define a map $\psi$ from $C$ to $\F_{q^k}=\F_q[x]/(Q(x))$ given by $\psi(\bc_f)=\bar{f}\in \F_q[x]/(Q(x))$. It is clear that $\psi$ is $\F_q$-linear. Thus, we can extend the domain of $\psi$ from $C$ to $\F_q^m$ so that it is an $\F_q$-linear map from $\F_q^m$ to $\F_{q^k}$.

As $(\pi(\F_{q^k}))^{\ast d}$ is a subcode of $C$,  it has  minimum distance at least $m-d(k-1)$.

Finally, we verify the following: for any $\bar{z}_1,\bar{z}_2,\dots,\bar{z}_d\in \F_q[x]/(Q(x))$ with $\deg(z_i(x))\le k-1$ for $1\le i\le d$, we have
\[\psi(\pi({\bar{z}_1})\ast\cdots\ast\pi({\bar{z}_d}))=\psi((z_1\cdots z_d)(\Ga_1),\dots,(z_1\cdots z_d)(\Ga_m))\\
=\overline{z_1z_2\cdots z_d}=\bar{z}_1\bar{z}_2\cdots \bar{z}_d.\]
This completes the proof.
\end{proof}

\section{Constructions of punctured Reed-Muller codes}\label{sec:RM}

In this section, we propose two constructions of punctured Reed-Muller codes. Both the constructions involve algebraic function fields over finite fields.
\subsection{Construction I}\label{subsec:6.1}
The first construction combines algebraic geometry codes with Sidon sequences. The idea of this construction is to replace variables $x_i$ in $\F_q[x_1,\dots,x_n]$ by suitable functions in a function field. The idea of replacing variables by functions of a function field was already proposed in  \cite{CGX12,CT13}. Furthermore, Sidon sequences were used to construct hitting sets and black box sets for polynomial identity testing in \cite{B12}. More precisely, in  \cite{B12}, variable $x_i$ is replaced by $x_i^{a_i}$, where $a_i$ is a number in the Sidon sequence. In this section, we combine the above two ideas to obtain explicit random punctured Reed-Muller codes.

\begin{lemma}\label{lem:6.1} Given a function field $K/F_q$ with $N+1$ distinct rational points $\Pin, P_1,P_2,\dots,P_N$, a Sidon sequence $\{a_i\}_{i=1}^n$ and functions $\{f_i\}_{i=1}^n$ of $K$ with $(f_i)_{\infty}=a_i\Pin$ for all $1\le i\le n$, then the punctured Reed-Muller code $\RM_{q}(n,d)|_{\mT_1}$ is a subcode of $C(\mP, dL\Pin)$ defined in Section  \ref{subsec:AG}  as long as $\max_{1\le i\le n}a_i\le L$, where $\mP=\{P_1,P_2,\dots,P_N\}$ and
$\mT_1:=\{(f_1(P_i),\dots,f_n(P_i)):\; i=1,2,\dots, N\}$. Thus,  $\RM_{q}(n,d)|_{\mT_1}$ has minimum distance at least $N-dL$. Furthermore, $\RM_{q}(n,d)|_{\mT_1}$ has  dimension ${n+d\choose n}$ if $N>dL$.
\end{lemma}
 \begin{proof}  It is easy to see that every codeword of $\RM_{q}(n,d)|_{\mT_1}$ is of the form $(F(f_1,\dots,f_n)(P_1),\dots,$\\  $F(f_1,\dots,f_n)(P_N))$ for a polynomial $F(x_1,\dots,x_n)\in\F_q[x_1,\dots,x_n]$.  Let $F=\sum_{\wt(I)\le d} \Ga_{I}X^I\in\F_{q}[X]$, where $X^I=x_1^{i_1}\cdots x_n^{i_n}$. Then  $F(f_1,f_2,\dots,$ $f_n)\in \mL(dL\Pin)$.
  Our first result follows.

 We consider the pole of $f_1^{i_1}\cdots f_n^{i_n}$. It is easy to see that $\nu_{\Pin}(f_1^{i_1}\cdots f_n^{i_n})=\sum_{k=1}^ni_ka_k$. Thus, $\nu_{\Pin}(f_1^{i_1}\cdots f_n^{i_n})\not=\nu_{\Pin}(f_1^{j_1}\cdots f_n^{j_n}) $ as long as $I\not=J$ and $\max\{\wt(I),\wt(J)\}\le d$. This implies that  $F(f_1,f_2,\dots,f_n)$ is a nonzero function of $\mL(dL\Pin)$ whenever $F(X)$ is a nonzero polynomial.

As $\RM_{q}(n,d)|_{\mT_1}$ is a subcode of $C(\mP, dL\Pin)$, the minimum distance $\RM_{q}(n,d)|_{\mT_1}$ is at least the minimum distance of $C(\mP, dL\Pin)$. The desired result on the minimum distance of $\RM_{q}(n,d)|_{\mT_1}$ follows from the Goppa minimum distance bound for $C(\mP, dL\Pin)$.

 Since $N>dL$, the corresponding $(F(f_1,\dots,f_n)(P_1),\dots,F(f_1,\dots,f_n)(P_N))$ is a nonzero vector. Therefore, the dimension of $\RM_{q}(n,d)|_{\mT_1}$ remains the same as $\RM_{q}(n,d)$, i.e., ${n+d\choose n}$.
 \end{proof}

By applying the Garcia-Stichtenoth tower and the Sidon sequnec in Corollary \ref{cor:3.3}, we obtain the following result.
\begin{cor}\label{cor:6.2} Let  $q=r^2$ for a prime power $r$ and let $M=n^d+O(n^{d-1})$. Put $t=\lceil\log_{{r}} (dM/2)\rceil$ and $N=r^{t}({r}-1)$. Then $\RM_{q}(n,d)|_{\mT_1}$ defined in Lemma \ref{lem:6.1} with the function field $F_t$ in the Garcia-Stichtenoth tower  is a $q$-ary $[N,{n+d\choose n}, \ge N-d(d+1)M]$-linear code.

In particular, if $d\le \Ge(\sqrt{q}-1)/2-1$  for some $0<\Ge<1$, then $\RM_{q}(n,d)|_{\mT_1}$ is a $q$-ary $[N,{n+d\choose n}]$-linear code with relative minimum distance at least $1-\Ge$ and rate
${{n+d} \choose n}/N \ge \Omega\left(\frac{\Ge}{d!}\right)$. Furthermore, one can deterministically construct $\RM_{q}(n,d)|_{\mT_1}$ in ${\rm poly}(n^d/\Ge)$ time.
 \end{cor}
\begin{proof}
Consider the Garcia-Stichtenoth tower $F_t/\F_{r^2}$ defined in Section \ref{subsec:GS} with $N(F_t)\ge r^{t}({r}-1)+1$  and $g(F_t)\le r^{t}-r^{t/2}$. Let $\Pin, P_1,P_2,\dots,P_N$ be $N+1$ distinct rational places of $F_t$.

Let $\{a_i\}_{i=1}^n$ be the Sidon sequence constructed in Corollary \ref{cor:3.3}. Then $dM\le \min_{1\le i\le n}a_i\le \max_{1\le i\le n}a_i< (d+1)M$.

Put $t=\lceil\log_{{r}} (dM/2)\rceil$.  Then $dM\ge 2g(F_t)$, For each $i\in[n]$, by the Riemann-Roch theorem, we have $\ell(a_i\Pin)=a_i+1-g(F_t)>(a_i-1)+1-g(F_t)=\ell((a_i-1)\Pin)$.
Thus,  one can find functions $f_i\in\mL((d+1)M\Pin)$ such that $\nu_{\Pin}(f_i)=-a_i$ for all $i=1,2,\dots,n$. Thus, the first statement of this corollary follows from Lemma \ref{lem:6.1}.

By choice of parameters, it is easy to verify that $N=O(n^d/\Ge)$, and the rate follows since ${n+d\choose n}\ge n^d/d!$.
 Hence, the relative minimum distance of $\RM_{q}(n,d)|_{\mT_1}$ is at least the one of $C(\mP,d(d+1)M\Pin)$ which is bigger than or equal to $1-d(d+1)M/N\ge 1-(d+1)/(r-1)\ge 1-\epsilon$.

 The construction of $\RM_{q}(n,d)|_{\mT_1}$ mainly involves: (i) construction of Sidon sequence $\{a_i\}_{i=1}^n$; (ii) finding of function $f_i$ with $\nu_{\Pin}(f_i)=-a_i$; and (iii) evaluation of functions at $N$ points $P_1,P_2,\dots,P_N$. By Lemma \ref{lem:sidon2},  $\{a_i\}_{i=1}^n$ can be constructed in ${\rm poly}(n^d)$ time. By Section \ref{subsec:GS}, $f_i$ with $\nu_{\Pin}(f_i)=-a_i$ can  be constructed in ${\rm poly}(n^d)$ multiplications in $\F_q$. Finally, evaluation of functions at $N$ points requires ${\rm poly}(N)={\rm poly}(n^d/\Ge)$ multiplications in $\F_q$. Since a multiplication in $\F_q$ requires $O((\log q)^2)$ operations, the desired result follows.
 \end{proof}

 \begin{example}\label{exm:1}{\rm In this example, we use the Hermitian function to illustrate our construction as a basis of the Riemann-Roch space of the Hermitian function field can be easily expressed.  Let $d=2$, $q=121$ and $n=10$. Then $\RM_{121}(10,2)$ has rate approximately equal to $9.8\times 10^{-20}$ and relative minimum distance $1-0.0165=0.9835$.
  Consider a $d$-Sidon sequence $1, 2, 4, 8, 13, 21, 31, 45, 66, 81$. By Lemma \ref{lem:sidon1}, the sequence $163, 164, 166, 170, 175, 183, 193,$ $ 207, 228, 243$ is also a  $d$-Sidon sequence.
 Consider the  function field $H=\F_{121}(x,y)$ defined by $y^{11}+y=x^{12}$. Then the genus $g$ of $H$ is $55$ and it has $11^3+1=1332$ rational places. Let $\Pin$ be the unique common pole of $x$ and $y$. Put $N=1331$. Then the multiset is
 \[\mT_1=\{(f_1(\Ga,\Gb),\dots,f_{10}(\Ga,\Gb)):\; \Ga,\Gb\in\F_{121},\; \Gb^{11}+\Gb=\Ga^{12}\},\]
 where $f_1=x^5y^9, f_2=x^4y^{10},f_3=x^{14}y, f_4=x^{10}y^5,f_5=x^5y^{10}, f_6=x^{9}y^7, f_7=x^{11}y^6,f_8=x^{9}y^{9}, f_9=x^{12}y^8, f_{10}=x^{21}y$. Then $\RM_{121}(10,2)|_{\mT_1}$ is a $121$-ary $[1331,66,\ge 845]$-linear code. Thus, this punctured code has rate and relative minimum distance approximately $0.0496$ and $1-0.365=0.635$, respectively.}
\end{example}

\subsection{Construction II} One constraint for the punctured Reed-Muller codes from Construction I is that $d$ is upper bounded by $O(\Ge\sqrt{q})$. In this section, we employ the multiplication friendly pairing to concatenate the code from Construction I to get $d=O(\Ge{q})$. The idea is to map a vector over $\F_{q^2}$ to several vectors over $\F_q$ through a multiplication friendly pairing. Thus, we obtain a punctured Reed-Muller code over $\F_q$ from algebraic geometry codes over $\F_{q^2}$. While in  Subsection \ref{subsec:6.1},  the punctured Reed-Muller code over $\F_q$ is obtained  from algebraic geometry codes over $\F_{q}$.

Assume that $d<q$, then by Lemma \ref{lem:5.2}
there exists a  $(d,2,q)_q$-multiplication friendly pair $(\pi,\psi)$ such that $(\pi(\F_{q^2}))^{\ast d}$ has relative minimum distance  at least  $1-d/q$. Note that the image $\pi(\F_{q^2})$ is a  $q$-ary $[q,2,q-1]$ Reed-Solomon code, i.e., $\pi$ is an  $\F_q$-isomorphism between $\F_{q^2}$ and a $q$-ary $[q,2]$-linear code.  For a column vector $\bv=(v_1,\dots,v_n)^T\in\F_{q^2}^n$, we obtain an $n\times q$ matrix
\[\pi(\bv)=\left(\begin{array}{c}
\pi(v_1)\\
\cdot\\
\cdot\\
\cdot\\
\pi(v_n)\end{array}
\right).\]
Denote by $\pi_i(\bv)$ the $i$-th column of $\pi(\bv)$ for $i=1,2,\dots,q$.

Then we have the following result.
 \begin{lemma}\label{lem:6.2a}
(see \cite[Theorem III.5]{GX14}) If there exists a $q^2$-ary extended punctured  Reed-Muller code $\RM_{q^2}(d,n)|_{\mT_1}$ with  parameters $[L, {n+d\choose d}, \ge L(1-\rho)]$ for some multiset $\mT_1\subseteq \F_{q^2}^n$, then the $q$-ary extended punctured  Reed-Muller code $\RM_q(d,n)|_{\mT_2}$ has  parameters $[qL, {n+d\choose d}, \ge qL(1-\rho-d/q)]$, where   $\mT_2$  is the multiset $\{\pi_i(\bu):\; 1\le i\le q, \bu\in \mT\}.$
 \end{lemma}
 The proof of the above lemma makes use of the multiplication property of the maps $\pi$ and $\psi$. Note that as $\left(\pi\left(\F_{q^2}\right)\right)^{*d}$ is a subcode of a $q$-ary $[q,d+1,q-d]$-linear code, the relative minimum distance of $\left(\pi\left(\F_{q^2}\right)\right)^{*d}$ is at least $1-d/q$. We refer to \cite[Theorem III.5]{GX14} for the detailed proof of the above lemma.

To define a multiset $\mT_2$, as in Lemma \ref{lem:6.1} one can consider an arbitrary  function field over $\F_{q^2}$ with certain points. For simplicity, we only consider the Garcia-Stichtenoth tower defined in Section \ref{subsec:GS} for our construction of $\mT_2$.

 \begin{cor}\label{cor:6.4}  $\RM_{q}(n,d)|_{\mT_2}$ is a $q$-ary $[N,{n+d\choose n}, \ge N(1-2(d+1)q^t/L-d/q)]$-linear code with length $N=qL=q^{t+1}(q-1)$.

In particular, if $d\le \Ge({q}-1)/4-1$  for some $0<\Ge<1$, then $\RM_{q}(n,d)|_{\mT_2}$ has the
 relative minimum distance at least $1-\Ge$, and rate
 ${n+d\choose n}/N \ge \Omega\left(\frac{\Ge^2}{d!}\right))$. Furthermore, one can deterministically construct $\RM_{q}(n,d)|_{\mT_2}$ in ${\rm poly}(n^d/\Ge)$ time.
 \end{cor}
 \begin{proof}   Take $q=r$ and consider the Garcia-Stichtenoth tower  $F_t/\F_{q^2}$ with $N(F_t)\ge q^{t}({q}-1)+1$  and $g(F_t)\le q^{t}-q^{t/2}$ as defined in Section \ref{subsec:GS}. Let $\Pin, P_1,P_2,\dots,P_L$ be $L+1$ distinct rational places of $F_t$, where $L=q^{t}({q}-1)$.

 Let $\{a_i\}_{i=1}^n$ be the Sidon sequence constructed in Corollary \ref{cor:3.3} and let $M=n^d+O(n^{d-1})$ be defined in  Corollary \ref{cor:3.3}. Put $t=\lceil\log_{{q}} (dM/2)\rceil$.  Then as shown in Subsection \ref{subsec:6.1}
  one can find functions $f_i\in\mL((d+1)M\Pin)$ such that $\nu_{\Pin}(f_i)=-a_i$ for all $i=1,2,\dots,n$. Put
  $\mT_2:=\{\pi_i((f_1(P_j),\dots,f_n(P_j))^T):\; 1\le i\le q,\; 1\le j\le L\}$ of $\F_q^n$.

 By choice of parameters, we have  $N=|\mT_2|=qL=q^{t+1}(q-1)=O(n^d/\Ge^2)$, and
 the rate is also clear since ${n+d\choose n}=\Omega(n^d/d!)$. As in the proof of Lemma \ref{lem:6.1}, we know that $F(f_1,f_2,\dots,f_n)$ is a nonzero function in $\mL(d(d+1)M\Pin)$ as long as $F(X)$ is a nonzero polynomial. Hence, the relative minimum distance of $\RM_{q^2}(n,d)|_{\mT_1}$ is at least the one of $C(\mP,d(d+1)M\Pin)$, where $\mT_1=\{(f_1(P_j),f_2(P_j),\dots,f_n(P_j)):\; j=1,2,\dots,L\}$ and $\mP=\{P_1,P_2,\dots,P_L\}$. By the minimum distance bound of algebraic geometry codes in Section \ref{subsec:AG}, the relative minimum distance of $C(\mP,d(d+1)M\Pin)$ is least $1-d(d+1)M/L\ge 1-2(d+1)q^t/L$. The first statement of this lemma follows from Lemma \ref{lem:6.2a}.

If $d\le \Ge({q}-1)/4-1$, then it is easy to verify that $1-2(d+1)q^t/L-d/q \ge 1-2(d+1)/(q-1)+\Ge/2\ge 1-\Ge$.  The desired result on parameters of $\RM_{q}(n,d)|_{\mT_2}$ follows.

Comparing with construction of $\RM_{q}(n,d)|_{\mT_2}$ in Corollary \ref{cor:6.2},  we need one more step for construction of $\RM_{q}(n,d)|_{\mT_2}$, namely concatenating $\RM_{q^2}(n,d)|_{\mT}$ with  $\pi(\F_{q^2})$. This can also be done in ${\rm poly}(q)={\rm poly}(d/\Ge)$ time. The proof is completed. \end{proof}

 \begin{example}\label{exm:2} We use the same Sidon sequence and function as in Example \ref{exm:1}. Let $\mT_1$ be the multiset given in Example \ref{exm:1}. By Lemma \ref{lem:6.2a}, $\RM_{11}(10,2)|_{\mT_2}$ is is a $11$-ary $[14641,66]$-linear code with the minimum distance at least $14641(1-722/1331-2/121)=6457$.
 \end{example}

\section{List decoding}
As both the constructions of punctured Reed-Muller codes are through concatenation of algebraic geometry codes, we need to recall some relevant results on list decodability of algebraic geometry codes first.

\subsection{List decoding of Reed-Solomon and algebraic geometry codes}
We begin with an extension of list decoding called list recovery. The notion was implicit in many works and was given this name in \cite{GI-focs01}.

\begin{defn} Let $\Sigma$ be the code alphabet of size $q$. A code $C\subseteq \Sigma^n$ is said to be {\it $(\tau, \ell, L)$-list recoverable} if for every family $\{S_i\}_{i=1}^n$ of subsets in $\Sigma$ with each subset of size at most $\ell$, the number of codewords $\bc=(c_1,\dots,c_n)\in C$ for which $c_i\notin S_i$ for at most $\tau n$ positions $i$ is at most $L$. Note that in the case $\ell=1$, this is the same as being $(\tau,  L)$-list decodable.
\end{defn}

Algebraic geometry codes have good list decodability \cite{G,GS99}. By
extending the results in \cite{GS99} from list decodability
to list recovery in a straightforward way, one obtains the following result.
  \begin{lemma}\label{lem:7.1} For small $\Ge>0$ and a positive integer $m$ with $m\le \Ge N$, where $N=r^t(r-1)$, the algebraic geometry code $C(\mP,(m+g)\Pin)$ based on the Garcia-Stictenoth tower $F_t/\F_{r^2}$ is $(1-\sqrt{\ell\Ge},\ell,O(1/\Ge))$-list recoverable, where $g=g(F_t)$. Furthermore, the decoding complexity is a polynomial in length $N$.
 \end{lemma}

Next, we recall a certain list decodability property of Reed-Solomon codes (see \cite{Sud1997,GS99} for complete version of list decoding of Reed-Solomon codes). 
\begin{lemma}\label{lem:7.2} For positive integers $\ell, d$ with $1\le \ell\le \sqrt{2q/(d+1)}$,  a $q$-ary  $[q, d+1,q-d]$ Reed-Solomon code $C$ is $\left(\tau,\ell\right)$-list decodable with $\tau\ge 1-\frac1{\ell+1}-\frac{\ell}2\times\frac{d}q$.
\end{lemma}

\subsection{List decoding $\RM_q(n,d)|_{\mT_1}$}
In this subsection, we consider list decoding of $\RM_q(n,d)|_{\mT_1}$ in Corollary \ref{cor:6.2}.
List decoding of  $\RM_q(n,d)|_{\mT_1}$ is straightforward since $\RM_q(n,d)|_{\mT_1}$ is a subcode of $C(\mP,d(d+1)M\Pin)$, where $C(\mP,d(d+1)M\Pin)$ is an algebraic geometry code from the function field $F_t$ defined in Section \ref{subsec:GS}. The length $N$ of  $\RM_q(n,d)|_{\mT_1}$ is $r^t(r-1)$. Put $q=r^2$. If $d\le \Ge(\sqrt{q}-1)/2-1$  for some $0<\Ge<1$,  then
\[\frac{d(d+1)M-g(F_t)}{N}\le\frac{2(d+1)}{r-1}\le {\Ge}.\]
By Lemma \ref{lem:7.1}, $\RM_q(n,d)|_{\mT_1}$ is $(1-\sqrt{\Ge},O(\frac1{\Ge}))$-list decodable. In conclusion, we have the following result.
\begin{theorem}\label{thm:7.3}  $\RM_{q}(n,d)|_{\mT_1}$ is a $q$-ary $[N,{n+d\choose n}]$-linear code with relative minimum distance at least $1-\Ge$ and rate $\Omega\left(\frac{\Ge}{d!}\right)$. Moreover, $\RM_{q}(n,d)|_{\mT_1}$ is $(1-\sqrt{\Ge},O(\frac1{\Ge}))$-list decodable and the decoding complexity is a polynomial in length $N$.
 \end{theorem}

\subsection{List decoding $\RM_q(n,d)|_{\mT_2}$}

In this subsection, we consider list decoding of $\RM_q(n,d)|_{\mT_2}$ in Corollary \ref{cor:6.4}.

\begin{theorem}\label{thm:7.4}  $\RM_{q}(n,d)|_{\mT_2}$ is a $q$-ary $[N,{n+d\choose n}]$-linear code with relative minimum distance at least $1-\Ge$ and rate $\Omega\left(\frac{\Ge^2}{d!}\right)$. Moreover, $\RM_{q}(n,d)|_{\mT_2}$ is $(1-{\Ge}^{1/3},O(\frac1{\Ge}))$-list decodable and the decoding complexity is a polynomial in length $N$.
 \end{theorem}
\begin{proof}
Let $(\pi,\psi)$ be a  $(d,2,q)_q$-multiplication friendly pair and put $\Sigma=(\pi(\F_{q^2}))^{*d}$. Then $\Sigma$ is a subcode of a $q$-ary $[q,d+1,q-d]$ Reed-Solomon code. Thus, $\RM_q(n,d)|_{\mT_2}$ is a code in $\Sigma^N$. In fact, $\RM_q(n,d)|_{\mT_2}$ is a concatenation of $C(\mP,d(d+1)M\Pin)$ and a $q$-ary $[q,d+1,q-d]$ Reed-Solomon code. More precisely speaking, the inner code is a subcode of a $q$-ary $[q,d+1,q-d]$ Reed-Solomon code, while the outer code is a subcode of $C(\mP,d(d+1)M\Pin)$ over $\F_{q^2}$. Since the inner code is $\left(1-\frac1{\ell+1}-\frac{\ell}2\times \frac{d}{q},\ell\right)$-list decodable and  the outer code $C(\mP,d(d+1)M\Pin)$ is $\left(1-\sqrt{\ell\Ge},\ell,O(\frac1{\Ge})\right)$-list recoverable, the concatenated code $\RM_q(n,d)|_{\mT_2}$ has decoding radius
 \[\tau=\left(1-\frac1{\ell+1}-\frac{\ell}2\times \frac{d}{q}\right)(1-\sqrt{\ell\Ge})\]
and list size $O(\frac1{\Ge})$  \cite{G}. Setting $\ell=\Theta(\Ge^{-1/3})$, we have $\tau=1-\Ge^{1/3}$.
\end{proof}

\section*{Acknowledgement}
\vspace{-1ex}
We are grateful to the two anonymous reviewers for their careful reading and helpful comments, which have substantially improved the presentation of the paper.

\end{document}